\documentclass[a4paper,pdftex,reqno,10pt]{amsart}
\usepackage{geometry}
\usepackage{graphicx}
\usepackage{multirow}
\usepackage{times}
\usepackage{qtree}
\usepackage{amssymb,amsmath}

\def\@begintheorem#1#2{\list{}{\thm@body}%
  \item[]{\bf #1~#2.}\quad\it\ignorespaces}
\def\@opargbegintheorem#1#2#3{\list{}{\thm@body}%
  \item[]{\bf #1~#2~\ifrembrks #3\global\rembrksfalse\else (#3)\fi.}%
  \quad\it\ignorespaces}
\def\@endtheorem{\endlist}
\newtheorem{theorem}{Theorem}[section]

\newtheorem{corollary}[theorem]{Corollary}
\newtheorem{lemma}[theorem]{Lemma}

\begin{document}

\title{Competitive learning of monotone Boolean functions}

\author{Sascha Kurz}
\address{University of Bayreuth, Department of Mathematics, D-95440 Bayreuth, Germany, sascha.kurz@uni-bayreuth.de}

\begin{abstract}
  We apply competitive analysis onto the problem of minimizing the number of queries to an oracle to completely reconstruct a given monotone Boolean function. Besides
   lower and upper bounds on the competitivity we determine optimal deterministic online algorithms for the smallest problem instances.
   
   \medskip
   
\end{abstract}

\maketitle

\section{Introduction}
\noindent
Complex systems with many components and resource restrictions can crash if certain combinations of the components are simultaneously active. As a demonstrative example one can imagine a computer system with different software packages. For $n$~software packages there are $2^n$ combinations where the computer can either crash or work properly. The behavior of the computer system can be described by a Boolean function $g:\{0,1\}^n\rightarrow \{0,1\}$. Suppose we have the possibility to evaluate the underlying Boolean function at arbitrary points. In practice we might think of asking an expert or performing an experiment like simply running the respective set of software packages. In our abstract setting we speak of asking a question. For an arbitrary Boolean function in any case $2^n$ questions are necessary (and sufficient if all questions are pairwise different) to unveil the entire function. Fortunately in many applications we can assume some restrictions. In our example it is quite reasonable to assume some kind of monotonicity. If the computer crashes for a certain subset $S\subseteq N:=\{1,\dots,n\}$ of the programs we can assume that it also crashes for every superset $N\supseteq T\supseteq S$ of the programs. Similarly, if the computer works properly for a set $S\subseteq N$ then it should also work properly for every subset $T\subseteq S$. So we restrict the underlying function to the class of monotone Boolean functions. For the ease of notation we write a Boolean function as $f:2^N\rightarrow \{0,1\}$ in the following, where $2^N$ denotes the set of subsets of~$N$. 

Since asking questions or performing experiments can be quite expensive one naturally tries to minimize the number of necessary questions. Different concepts like worst case or average case analysis have been applied on this problem so far. As shown shown by Engel \cite{0868.05001} up to ${n \choose {\left\lfloor\frac{n}{2}\right\rfloor}}+{n \choose {\left\lfloor\frac{n}{2}\right\rfloor+1}}$ questions are necessary to uniquely verify the worst case examples. There are algorithms, see e.g.\ \cite{Hansel}, which achieve this unavoidable worst case bound. Since there are monotone Boolean functions which can be uniquely verified asking a single question, those, with respect to worst case analysis, optimal algorithms might not be adequate in all practical applications. Thus there are studies in the literature minimizing the average number of necessary questions while assuming an uniform distribution of the possible function, see e.g.\ \cite{Torvik1}. In this paper we want to study exact learning of monotone Boolean functions using competitive analysis. This concept has the big advantage that no assumptions on the distribution of the occurring functions are necessary.

\section{Preliminaries}
\label{sec_preliminaries}

\noindent
We call a function $f:2^N\rightarrow\{0,1\}$ a monotone Boolean function if $f(S)=1$ implies $ f(T)=1$ for all $N\subseteq T\subseteq S$ and $f(S)=0$ implies $f(T)=0$ for all $T\subseteq S\subseteq N$. The set $\mathcal{L}$ of the maximal lower sets consists of the sets $S\subseteq N$ with $f(S)=0$ where $f(T)=1$ for all proper supersets of $S$. Similarly the set $\mathcal{U}$ of the minimal upper sets consists of the sets $S\subseteq N$ with $f(S)=1$ where $f(T)=0$ for all proper subsets of $S$. We would like to remark that either $\mathcal{L}$ or $\mathcal{U}$ suffice to uniquely characterize $f$ within the class of Boolean functions, i.e.\ we have $f(S)=1$ if and only if there exists an $U\in\mathcal{U}$ with $U\subseteq S$ or $f(T)=0$ if and only if there exists an $L\in\mathcal{L}$ with $T\subseteq L$.

Via asking $f(S)$ we do not get the information whether $S$ is an inclusion-maximal lower or an inclusion-minimal upper set but only $f(S)=1$ or $f(S)=0$. So in general asking the sets from $\mathcal{L}$ or $\mathcal{U}$ is not sufficient. 

\begin{lemma}
  A monotone Boolean function $f$ is uniquely characterized if and only if all values of $\mathcal{L}$ and $\mathcal{U}$ are given. \\[-7mm]
\end{lemma}
\begin{proof}
  Suppose there is a single $U\in\mathcal{U}$ whose value is not known. We construct a Boolean function $f'$ by setting $f'(S)=f(S)$ for all subsets $S\neq U$ and
  $f'(U)=0$. This function is also monotone since for all proper subsets $T\subsetneq U$ we have $f'(T)=f(T)=0$. If there is a single $L\in\mathcal{L}$ whose value is
  not known we can consider the monotone Boolean function $f''$ with $f''(S)=f(S)$ for all $S\neq L$ and $f''(L)=1$.\hfill{$\square$}
\end{proof}

So let us denote by $m(f)$ the cardinality $\left|\mathcal{U}\cup\mathcal{L}\right|$, i.e.\ a lower bound for each deterministic algorithm to reconstruct $f$ via asking questions. As notation for those algorithms we use binary decision trees, where we use the questions as nodes and where the answer ``$0$'' corresponds to the left successor:

\smallskip

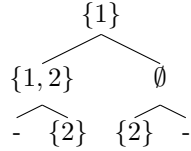
\Tree [.$\{1\}$ [.$\{1,2\}$ - $\{2\}$ ] [.$\emptyset$ $\{2\}$ [.- ] ] ]

\medskip

\noindent
In this example for $n=2$ the first question is $\{1\}$. Once the answer is $0$ the algorithm continues by asking $\{1,2\}$. If the answer is again $0$ then there is only one possible monotone Boolean function (the  all-zero function $f_0$: $f_0(S)=0$ for all $S\subseteq N$) left and we have reconstructed $f$. By $A(f)$ we denote the number of questions asked by algorithm $A$ to reconstruct $f$. In our example we have $A(f)=2$ while $m(f)=1$ would be possible for the optimal algorithm. In competitive analysis the fraction $\frac{A(f)}{m(f)}$ is studied.

We call a deterministic algorithm~$A$ reconstructing an $n$-variable monotone Boolean function $c$-competitive for a real number $c\ge 1$ if
$\frac{A(f)}{m(f)}\le c$ for all $f\in\mathcal{M}_n$, where $\mathcal{M}_n$ denotes the set of monotone Boolean functions on $n$ variables.
The best possible competitivity is denoted by $c_n^\star$, i.e.\ the infimum of the possible $c$ for $c$-competitive algorithms on $n$-variables.

Let us at first comment on the competitivity of some classical learning algorithms. The Hansel's algorithm \cite{Hansel} behaves very badly using this measure, e.g.\ for the all-one function $f_1$ ($f_1(S)=1$ for all $S\subseteq N$) with $m(f_1)=1$ at least ${n \choose {\left\lfloor\frac{n}{2}\right\rfloor}}$ questions are asked (there is some freedom in the definition of Hansel's algorithm), see e.g.\ \cite{Torvik1}. Thus Hansel's algorithm is not $c$-competitive for $c<{n \choose {\left\lfloor\frac{n}{2}\right\rfloor}}$ while being worst-case optimal.

Another algorithm for learning a monotone Boolean function is the so called \textsc{Find-Border} algorithm of Gainanov \cite{find_border} (which is used as a subroutine in several other learning algorithms). In each iteration an element of $\mathcal{U}$ is determined and verified using at most $n+1$ questions. Thus the \textsc{Find-Border} is $n+1$-competitive for all $n\in\mathbb{N}$. We would like to remark that a refined analysis shows that at most $n\cdot|\mathcal{U}|+1+|\mathcal{L}|$ questions are asked and that it can be slightly adopted to yield an $n$-competitive algorithm for $n\ge 2$. (If the elements of $\mathcal{L}$ are iteratively determined then at most $n\cdot|\mathcal{L}|+1+|\mathcal{U}|$ questions are asked.)

The enumeration of the set $\mathcal{M}_n$ is a classical combinatorial problem known as Dedekind's problem \cite{28.0186.04}. So far the exact numbers could be determined only up to $n=8$ and are given by $3$, $6$, $20$, $168$, $7\,581$, $7\,828\,354$, $2\,414\,682\,040\,998$, and $56\,130\,437\,228\,687\,557\,907\,788$, see e.g.\ \cite{1072.06008}. To factor out symmetry we call two monotone Boolean functions $f$ and $g$ equivalent if there is a bijection $\sigma$ on $N$ such that $f(S)=g(\sigma(S))$ for all $S\subseteq N$. The number of inequivalent monotone Boolean functions (or orbits) are given by $3$, $5$, $10$, $30$, $210$, $16353$, see e.g.\ \cite{OEIS} , but grow nevertheless double exponentially.
 
\section{Lower bounds on the optimal competitivity}

\noindent
Based on the fact that the all-zero function $f_0$ and the all-one function $f_1$ need only one question to be completely reconstructed, we can state $c_n^\star\ge 2$ for all $n\in\mathbb{N}$. By $b_i(n)$ we denote the number of monotone Boolean functions on $n$ variables with $m(f)=i$. Since the answer to each question splits the set of monotone Boolean functions which are compatible with the answers so far into two subsets of remaining candidates, we have:

\begin{lemma}
  $$c_n^\star\ge \frac{\left\lceil\log_2\left(\sum\limits_{j=1}^{i} b_j(n)\right)\right\rceil}{i}\quad\forall i\ge 1.$$
\end{lemma}

\begin{lemma}
  $$ b_{i+1}(n)\ge {n\choose i}\quad \forall 1\le i\le n.$$
\end{lemma}
\begin{proof}
  Let $S$ be an arbitrary $i$-element subset of $n$. For the monotone Boolean function with unique maximal lower set $N\backslash S$ the minimal upper sets correspond
  to the elements of $S$.\hfill{$\square$}
\end{proof}

\begin{corollary}
  For each $\varepsilon>0$ there is a $n_0(\varepsilon)$ such that $c_n^\star\ge(1-\varepsilon)\log_2 n$ for all $n\ge n_0(\varepsilon)$.
\end{corollary}

\begin{lemma}
  If $U\in\mathcal{U}$ then $|\mathcal{L}|\ge |U|$.
\end{lemma}
\begin{proof}
  Each $L\in\mathcal{L}$ can at most contain one set $U\backslash\{i\}$ with $i\in U$ as a subset.\hfill{$\square$}
\end{proof}

Using this one can easily determine $b_1(n)=2$, $b_2(n)=n$, $b_3(n)=2{n\choose 2}$, and $b_4(n)=8{n\choose 3}$. 
%For asymptotic results concerning the number of monotone Boolean functions with prescribed $|\mathcal{L}|$ we refer the interested reader to~\cite{MR1935742}.

\begin{lemma}
 $$c_{n+1}^\star\ge c_n^\star.$$
\end{lemma}
\begin{proof}
  Let $A$ be a deterministic online algorithm for $n+1$ variables. We can obtain an online algorithm $A'$ for $n$ variables by adjusting each question $S$ to
  $S\backslash\{n+1\}$ and slightly adapting the final output.
  
  Let $f$ be an arbitrary monotone Boolean function on $n$ variables and $g$ be a monotone Boolean function on $n+1$ variables defined via
  $f(S)=g(S)=g(S\cup\{n+1\})$ for all $S\subseteq \{1,\dots,n\}$. Due to $\mathcal{U}(g)=\{U\mid U\in\mathcal{U}(f)\}$ and
  $\mathcal{L}(g)=\{L\cup\{n+1\}\mid L\in\mathcal{L}(f)\}$ we have $m(g)=m(f)$.\hfill{$\square$}
\end{proof}

\section{Optimal algorithms for small $n$}

\noindent
We call a deterministic online algorithm to reconstruct a monotone Boolean function or its corresponding binary decision tree reasonable if only sets are asked whose function value cannot be deduced from previous answers. For $n=1$ variable there are only two reasonable binary decision trees, each having a competitivity of $c_1^\star=2$. For $n=2$ variables an example with competitivity $c_2^\star=2$ is given in Section~\ref{sec_preliminaries}.

Our next aim is to prove that every deterministic online algorithm for $n=3$ variables has a competitivity of at least $\frac{5}{2}$. To conclude a lower bound on the competitivity it suffices to give a sequence of answers to the questions of the algorithm that are compatible with a monotone Boolean function. By choosing a suitable sequence of answers for a given deterministic algorithm we can conclude the tight lower bound. Since we use only a path of the binary decision tree, consisting of the sequence on questions, the same sequence of answers results in the same lower bound for a large set of binary decision trees. We can further reduce the set of candidates of paths by utilizing symmetry. Therefore we call two such paths $P_1=(S_1,\dots,S_l)$ and $P_2=(T_1,\dots,T_l)$, where $S_1,\dots,S_l,T_1,\dots, T_l\subseteq N$, equivalent if there is a bijection $\sigma$ of $N$ fulfilling $\sigma(S_i)=T_i$ for all $1\le i\le l$. It suffices to consider inequivalent paths only, e.g.\ we can assume that the first question is either $\emptyset$, $\{1\}$, $\{1,2\}$, or $\{1,2,3\}$.

Suppose that the answer to this first question $S_1$ is one if $|S_1|\ge 2$ and zero otherwise, then $S_2\in\{\emptyset,N\}$ since otherwise we could not exclude the all-one function $f_1$ or the all-zero function $f_0$ and would end up with an algorithm having a competitivity of at least $3$. We can abstain from further considering the initial path segments $(\emptyset,N)$ and $(N,\emptyset)$ since the initial path segments $(\{1\},N)$ and $(\{1,2\},\emptyset)$ yield more information about the unknown Boolean function. In Figure~\ref{fig_lower_n3} we depict the remaining part of our argument graphically. Each vertex is labeled with question and answer. At the leafs we further specify a compatible monotone Boolean function with $m(f)=2$ via $\mathcal{U}$ and $\mathcal{L}$. Since each leaf has height $3$ and the corresponding elements of $\mathcal{U}$ and $\mathcal{L}$ have not been asked so far, $\frac{3+2}{2}$ is a lower bound for the competitivity in each case. 

\begin{figure}[htp]
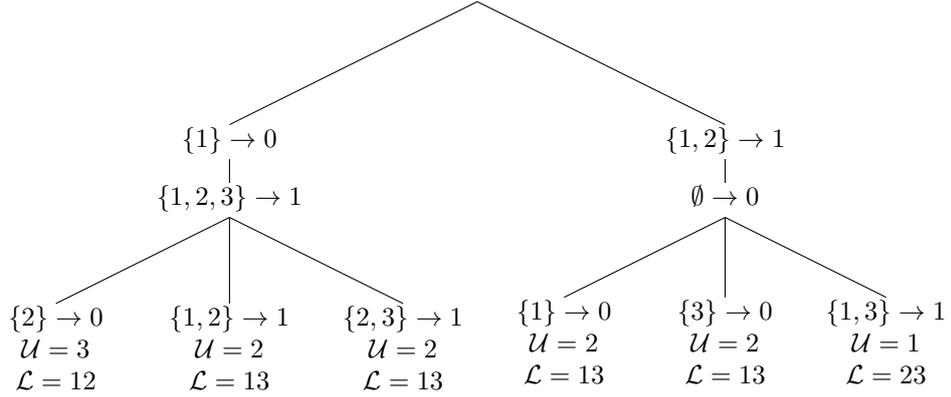

\begin{center}
\Tree [ [.$\{1\}\rightarrow 0$ [.$\{1,2,3\}\rightarrow 1$ $\begin{array}{c}\{2\}\rightarrow 0\\\mathcal{U}=3\\\mathcal{L}=12\end{array}$ $\begin{array}{c}\{1,2\}\rightarrow 1\\\mathcal{U}=2\\\mathcal{L}=13\end{array}$ $\begin{array}{c}\{2,3\}\rightarrow 1\\\mathcal{U}=2\\\mathcal{L}=13\end{array}$ ] ] [.$\{1,2\}\rightarrow 1$ [.$\emptyset\rightarrow 0$ $\begin{array}{c}\{1\}\rightarrow 0\\\mathcal{U}=2\\\mathcal{L}=13\end{array}$ $\begin{array}{c}\{3\}\rightarrow 0\\\mathcal{U}=2\\\mathcal{L}=13\end{array}$ $\begin{array}{c}\{1,3\}\rightarrow 1\\\mathcal{U}=1\\\mathcal{L}=23\end{array}$ ] ] ]
\caption{Lower bound for three variables.}
\label{fig_lower_n3}
\end{center}
\end{figure}

For the other direction we describe a whole class of $\frac{5}{2}$-competitive online algorithms in Figure~\ref{fig_upper_n3}. As for a binary decision tree the internal vertices are labeled with the questions of the algorithm. The leafs are either labeled with $[c]$ or $[u,k,c]$. In the first case there is only a unique monotone Boolean function being compatible with the previous answers left, so that this path is $c$-competitive. In the later case there are only $u$ unclassified sets, i.e.\ sets whose value cannot be deduced from the previous answers and the minimum $m(f)$ of the remaining compatible monotone Boolean functions is $k$ so that every reasonable continuation of the online algorithm is $c$-competitive in the subtree starting at this leaf. Thus we have $c_3^\star=\frac{5}{2}$.

\begin{figure}
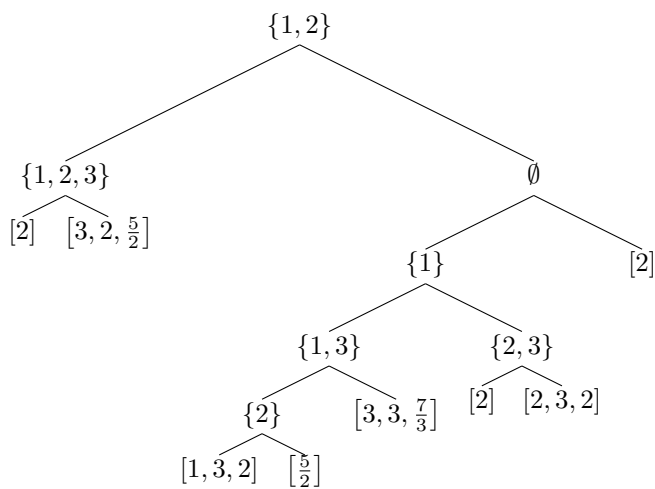

\begin{center}
\Tree [.$\{1,2\}$ 
        [.$\{1,2,3\}$ $[2]$ $\left[3,2,\frac{5}{2}\right]$ ]
        [.$\emptyset$ 
          [.$\{1\}$ 
            [.$\{1,3\}$ 
              [.$\{2\}$ $[1,3,2]$ $\left[\frac{5}{2}\right]$ ] 
              $\left[3,3,\frac{7}{3}\right]$ 
            ] 
            [.$\{2,3\}$ $[2]$ $[2,3,2]$ ] 
          ]
          $[2]$
        ]
      ]
\caption{A class of $\frac{5}{2}$-competitive algorithms for three variables.}
\label{fig_upper_n3}
\end{center}
\end{figure}

\medskip

Using the same ideas and larger trees one can show $c_4^\star=\frac{8}{3}$ and $c_5^\star\ge 3$.

\section{Conclusion}

\noindent
We have considered the problem of minimizing the number of questions to an oracle to completely reconstruct an unknown monotone Boolean function from the perspective of competitive analysis. The classical algorithm of Hansen turns out to perform pretty bad using this measure. For given general monotone Boolean functions bounds on the best possible competitivity are far from being tight. As shown in \cite{1072.06008} a \textit{typical} monotone Boolean function, i.e.\ almost all of $\mathcal{M}_n$ functions, fulfills $m(f)\ge|\mathcal{L}|\ge \frac{1}{2}{n\choose \left\lceil\frac{n}{2}\right\rceil}-n2^{n/2}$. Thus most reasonable algorithms have a constant competitivity on almost all inputs. So the challenge is to deal with those monotone Boolean functions with atypically \textit{small} $m(f)$. On the other hand the ratio between $|\mathcal{L}|$ and $|\mathcal{U}|$ can become exponential, see \cite{10.1007_BF00116828}.

Here we have determined optimal algorithms for rather small problem instances only. On the other hand the described methods and shortcuts may be used in order to implement a non-trivial search to determine the next exact values of $c_n^\star$. (This is indeed what we plan to do next.) A direct exhaustive search on all reasonable binary decision trees seems impracticable even for rather small~$n$. 

%% \bibliography{learning}
%% \bibdata{learning}
%% \bibliographystyle{amsplain}

\end{document}